\documentclass[10pt]{article}
\usepackage[fleqn]{amsmath}
\usepackage{amsmath,amsthm,amsfonts,amscd,bezier,caption,braket,MnSymbol}
\DeclareMathOperator{\Hom}{Hom}
\DeclareMathOperator{\Com}{Com}
\DeclareMathOperator{\CUP}{Cup}
\DeclareMathOperator{\II}{I}
\DeclareMathOperator{\1}{id}

\DeclareMathOperator{\Ker}{Ker}
\DeclareMathOperator{\IM}{Im}
\DeclareMathOperator{\ch}{char}
\renewcommand{\=}{:=}
\renewcommand{\t}{\otimes}
\renewcommand{\u}{\smile}
\renewcommand{\=}{:=}
\renewcommand{\t}{\otimes}
\renewcommand{\o}{\circ}

\renewcommand{\:}{\colon}
\newcommand{\ZZ}{\mathbb{Z}}
\newcommand{\CC}{\mathbb{C}}
\newcommand{\EEnd}{\mathcal End}
\newcommand{\EE}{\mathcal E}

\newcommand{\bul}{\circ}
\newcommand{\f}{f}

\newcommand{\h}{h}
\newcommand{\w}{\omega}
\newcommand{\p}{d}%{\partial}

\newcommand{\NN}{\mathbb{N}}
\newcommand{\de}{\delta}
\newtheorem{thm}{Theorem}[section]

\newtheorem{prop}[thm]{Proposition}

\newtheorem{cor}[thm]{Corollary}
 
\theoremstyle{definition}
\newtheorem{defn}[thm]{Definition}

\theoremstyle{definition}
\newtheorem{exam}[thm]{Example}

\theoremstyle{definition}
\newtheorem{rem}[thm]{Remark}

\theoremstyle{definition}
\newtheorem{disc}[thm]{Discussion}

\numberwithin{equation}{section}
\numberwithin{table}{section}

\theoremstyle{theorem}
\newtheorem*{MOD1}{MOD I}
\newtheorem*{MOD2}{MOD II}
\newtheorem*{MOD3}{MOD III}
\begin{document}
\title{\Large\bf The operadic modeling of gauge systems\\of the Yang-Mills type}
\date{}
\author{Eugen Paal
}

\maketitle
\thispagestyle{empty}
\begin{abstract}
The basics of operadic variational formalism is presented which is necessary when modeling the operadic systems. A  general gauge theoretic approach to the abstract operads, based on the physical measurements concepts,  is justified and considered. It is explained how  the matrix and Poisson algebra relations can be extended to operadic realm. The tangent cohomology spaces of the binary associative flows with their Gerstenhaber algebra  structure can be seen as equally natural objects for operadic modeling, just as the matrix and Poisson algebras in conventional modeling.  In particular, the relation of the tangent  Gerstenhaber algebras to operadic Stokes law for operadic observables is revealed and discussed. Based on this, the \textit{rational (cohomological) variational principle} and \textit{operadic Heisenberg equation} for quantum operadic flows are stated. As a modeling selection rule,  the  operadic gauge equations of the Yang-Mills type are considered and justified from the point of view of the physical measurements and the algebraic deformation theory. It is also shown how the binary weakly non-associative operations are related to approximate operadic (anti-)self-dual models.
\end{abstract}

%%%%%
\section{Introduction and outline of the paper}

It is well known how the Lie bracketing for matrices $f$ and $g$ is related to the matrix multiplication via  the (left and right) Leibniz rules as follows:
\[
 [h,fg]=[h,f]g+f[h,g], \quad 
 [fg,h]=f[g,h]+[f,h]g.
\]
Similar relations hold in the Poisson algebras that are everywhere used in (classical and quantum) physics. In a sense, the \textit{operad algebra} is a natural extension of the matrix and Poisson algebras to operadic realm, thus providing us with a natural well defined generalized (differential) calculus called the \textit{operad calculus}. 

Here it must be recalled the importance of the Leibniz rule - this is inevitably used for \textit{estimations} of the \textit{physical measurements errors}. Thus,  a general \textit{gauge} theoretic approach to operadic systems,  based on the physical measurements concepts, is needed. 

At first, let us list some related references.

In 1963, Gerstenhaber discovered \cite{Ger63} an operadic (pre-Lie) system in the Hochschild complex of an associative algebra  and used it to study of the algebraic structure of the cohomology of an associative algebra. Among others, Gerstenhaber proved a variant of the \textit{operadic Stokes  law} (see Sec. \ref{Stokes} for details) for the Hochschild cochains, which implies the Leibniz rule in the Hochschild cohomology. The notion of a symmetric operad was fixed and more formalized by May \cite{May72} as a tool for iterated loop spaces. In 1995 \cite{GV95a,GV95b}, Gerstenhaber and Voronov listed the main principles of the \textit{brace formalism} for the  \emph{operadic flows} in the Hochschild complex. Quite a remarkable research activity in the operad theory and its applications can be observed in the last decades, see e.g \cite{LV12,Fresse09,MaShSta01,Smi01,Rene97} and quite extensive  bibliographies therein.

In this paper, the basics of operadic variational formalism is presented which is necessary when modeling the operadic systems.  It is explained how  the matrix and Poisson algebra relations can be extended to operadic realm. The tangent cohomology spaces of binary associative operations with their Gerstenhaber algebra  structure can be seen as equally natural objects for operadic modeling, just as the matrix and Poisson algebras in conventional modeling.    In particular, the relation of the tangent Gerstenhaber algebras to operadic Stokes law for operadic flows is revealed and discussed. Based on this, the \textit{rational (cohomological) variational principle} and \textit{operadic Heisenberg equation} for quantum observable flows are stated. As a modeling selection rule,  the  \textit{operadic gauge equations} of the Yang-Mills type are considered and justified from the point of view of the physical measurements and the algebraic deformation theory. It is also shown how the binary weakly non-associative operations are related to approximate operadic (anti-)self-dual models.

\section{Operadic (composition) system}

Let $K$ be a unital associative commutative ring and let $C^n$ ($n\in\NN$) be unital $K$-modules. For  a (homogeneous) $f\in C^n$, we refer to $n$ as the \textit{degree} of $f$ and often write (when it does not cause confusion) $f$ instead of $\deg f$. For example, $(-1)^f\=(-1)^n$, $C^f\=C^n$ and $\o_f\=\o_n$. Also, it is convenient to use the \textit{reduced} (desuspended) degree $|f|\=n-1$. Throughout this paper, we assume that $\t\=\t_K$ and  $\Hom\=\Hom_K$. Sometimes, to simplify the presentation, we silently assume $\ch K\neq 2, 3$ or $K=\CC$.

\begin{defn}[BAG, Fig. \ref{fig_bag}, well known]
Define the maps $B,A,G: C^{h}\t C^{f} \mapsto\NN\times\NN$ sending $h\otimes f$ to
\begin{align*}
B(h\t f)&\= \braket{(i,j)\in\NN\times\NN \, | \,  1\leq i \leq |h|; \, 0\leq j \leq i-1},\\
A(h\t f)&\= \braket{(i,j)\in\NN\times\NN \, | \, 0\leq i \leq |h|; \, i\leq j \leq i+|f|},\\
G(h\t f)&\=\braket{(i,j)\in\NN\times\NN \, |\,  0\leq i \leq |h|-1;  \, i+f \leq j \leq |f|+|h|}.
\end{align*}
One can represent an image of $BAG\=B\sqcup A\sqcup G$ as a discrete rectangle in Figure \ref{fig_bag}, it is the \textit{disjoint} union of the triangles $B,G$ and parallelogram $A$.
\end{defn}

%\begin{center}
\begin{figure}[h]
\caption{$BAG$}
\setlength{\unitlength}{5mm}
\begin{center}
%\begin{picture}(12,17)(-5,-3)
\begin{picture}(5,12)(0,0)
\linethickness{.45pt}
\qbezier(1,0)(3,0)(6,0)
\put(-0.2,-1){$0$}
\put(0.9,-1){$1$}
\put(5.7,-1){$|{\h}|$}
\qbezier(6,0)(6,2,5)(6,5)
\qbezier(1,0)(3.5,2.5)(6,5)
\put(5,0.5){${B}$}
\qbezier(0,0)(0,2)(0,4)
\put(-1,3.7){$|{\f}|$}
\qbezier(0,0)(3,3)(6,6)
\qbezier(6,6)(6,8)(6,10)
\qbezier(0,4)(3,7)(6,10)
\put(2.75,4.75){${A}$}
\qbezier(0,5)(0,5)(0,10)
\put(-0.7,4.8){${\f}$}
\put(-3,9.7){$|{\f}|+|{\h}|$}
\qbezier(0,5)(2.5,7.5)(5,10)
\qbezier(0,10)(2.5,10)(5,10)
\put(0.5,8.8){${G}$}
\put(7,0){\vector(1,0){1.5}}
\put(7.7,-1){$i$}
\put(0,11){\vector(0,11){1.5}}
\put(-0.7,11.5){$j$}
\end{picture}
\end{center}
\label{fig_bag}
\end{figure}
%\end{center}

\begin{defn}
A (right) linear (pre-)\textit{operad} (operadic or composition system, non-symmetric operad, non-$\Sigma$ operad etc) with coefficients in $K$ is a sequence $C\=(C^n)_{n\in\NN}$ of unital $K$-modules (an $\NN$-graded $K$-module), such that for all $m\in\NN$ the following conditions hold:
\begin{enumerate}
\item[(1)]
For $0\leq i\leq m-1$ there exist linear maps called the (partial) \textit{compositions}
\[
  \o_i\in\Hom(C^m\t C^n,C^{m+n-1}),\quad |\o_i|=0.
\]
\item[(2)]
For all $h\t f\t g\in C^h\t C^f\t C^g$, the \textit{composition (associativity) relations} hold,
\[
(h\o_i f)\o_j g=
\begin{cases}
    (-1)^{|f||g|} (h\o_j g)\o_{i+|g|}f  &\text{if $(i,j)\in B$}\\
    h\o_i(f\o_{j-i}g)  &\text{if $(i,j)\in A$}\\
    (-1)^{|f||g|}(h\o_{j-|f|}g)\o_i f  &\text{if $(i,j)\in G$}                       
\end{cases}
\]
\item[(3)]
There exists a \textit{unit} $\II\in C^1$, also  called a \textit{root}, such that
\[
\II\o_0 f=f=f\o_i \II,\quad 0\leq i\leq |f|.
\]
\end{enumerate}
The homogeneous elements of $C$ are called  \textit{operations}.
\end{defn}

\begin{rem}
In item (2), subitems $B$ and $G$ turn out to be equivalent.
\end{rem}

\begin{rem}
One may collect the sequence of compositions as $\hat{\o}\=(\o_n)_{n\in\NN}$. The associativity of the latter has been exhaustively  explained in \cite{LV12,MaShSta01,Rene97}.
\end{rem}

\begin{exam}[endomorphism operad {\cite{Ger63}}]
\label{HG} 
Let $L$ be a unital $K$-module and
$\EE_L^n\={\EEnd}_L^n\=\Hom(L^{\t n},L)$. Define the partial compositions
for $f\t g\in\EE_L^f\t\EE_L^g$ by 
\[
f\o_i g\=(-1)^{i|g|}f\o(\1_L^{\t i}\t g\t\1_L^{\t(|f|-i)}),  \quad 0\leq i\leq |f|.
\]
Then $\EE_L\=(\EE_L^n)_{n\in\NN}$ is an operad with the unit $\II\=\1_L\in\EE_L^1$ called the  \textit{endomorphism operad} of $L$.  Thus, in particular, the algebraic operations can be seen as elements of an endomorphism operad. This motivates using the term \textit{operations} for homogeneous elements of an abstract operad.
\end{exam}

\begin{exam}
(1) Planar rooted trees, (2) cylindrical (biological) trees (3) strings, (4) little squares and disks etc (5) the dynamical growth  (change) of biological trees. One can find many examples with exhaustive explanations in \cite{LV12,MaShSta01,Rene97,May72} and references therein.
\end{exam}

\begin{defn}[representations]
A linear map $\Psi\in\Hom(C,\EE_L)$ is called a \textit{representation} of $C$ if
\[
\Psi_f \o_i \Psi_g = \Psi_{f\o_i g}, \quad i= 0,\ldots,|\psi_f| :=|f|.
\]
\end{defn}

In this restricted sense, \textit{operad algebra} means a representation of an operad. One may also say that algebras are representations of operads, the construction is similar to the concept of representation (realization) of groups and associative algebras. In modeling problems using representations and modules over operads is inevitable.

\medskip
In what follows we omit various associated prefixes, such as "pre-", "non-symmetric", "symmetric", "endomorphism" etc, and to cover a wider context we often use the flexible term  \textit{operadic system}.

\section{Operadic flows}

Now let $\mu\in C^2\subset C$ and call it a \textit{binary} operation, it may be considered as an elementary element, yet without any interior structure. Following Gerstenhaber \cite{Ger63} and \cite{KPS00,KP01, KP02},  we define the low-order \textit{ground simplexes} and \textit{operadic flows} associated with a given operadic system,  and start listning their basic properties.
\begin{defn}[ground simplexes]
%\label{ground-simplexes}
Define the Gerstenhaber (discrete) \textit{ground simplexes} as follows: 
\begin{align*}
&\braket{h}= \braket{i \in \NN\, | \, 0 \leq i \leq |h|}, \\
&\braket{hf}\= \braket{(i,j) \in \NN^{2}\, | \, 0\leq i \leq |h|-1; \, i+f\leq j \leq |f|+|h|}=G(h\t f) ,\\
&\braket{hfg} \= \braket{(i,j,k) \in \NN^{3}\, | \, 0\leq i \leq |h-2 |; \, i+f \leq j \leq |h|+|f|-1; j+g\leq k \leq |h|+|f|+|g|}.
\end{align*}
We denote the (conventional simplicial) \emph{boundary operator} acting on the ground simplexes by  $\partial$.
\end{defn}

\begin{defn}[operadic flows]
The low-order  (simplicial) \textit{operadic flows} are defined as (pairing) superpositions over the corresponding ground simplexes, 
\begin{align*}
&\braket{h|f} \=h\bul f \= \sum_{\braket{h}} h\o_i f \quad \in C^{h+|f|}, \\
&\braket{h|fg} \= \braket{h|fg} \=\sum_{\braket{hf}}(h\o_i f)\o_j g \quad \in C^{h+|f|+|g|}, \\
&\braket{h|fgb} \= \braket{h|fgb}\=\sum_{\braket{hfg}}  ((h\o_{i}f)\o_{j}g)\o_{k}b \quad \in C^{h+|f|+|g|+|b|}.
\end{align*}
One can see that $|\braket{\cdot|\hphantom{\cdot}}|=0$. Evidently, every operation $f$ can be presented as the flow $f=\braket{\II|f}$. Sometimes it is useful to expose the number of ket-arguments, called the \textit{order} of an operadic flow, e.g, we may use  $\braket{\cdot|\cdot}\=\o$, $\braket{\cdot|\cdot\cdot}$,  $\braket{\cdot|\cdots}$ or similar notations. The pair $\Com C\=(C,\bul)$ is called the \textit{composition algebra} of $C$ with the \textit{total composition} $\o$. 
\end{defn}

\begin{rem}
The higher-order operadic  flows can be easily seen (formally generated) from superpositions of the planar rooted trees or, more sophistically, by using compositions in an endomorphism operad \cite{GV95a,GV95b}. In this paper, we do not use the higher-order flows.
\end{rem}

\begin{defn}[cup]
The \textit{cup-multiplication} $\u_\mu \: C^f\t C^g\to C^{f+g}$ is defined by
\[
f\u_\mu g\=(-1)^f(\mu\o_0 f)\o_f g \quad \in C^{f+|g|}
, \quad 
|\smile_\mu|=1.
\]
The pair $\CUP_\mu C\=(C,\smile_\mu)$ is called a $\u_\mu$-\textit{algebra} (\textit{cup-algebra}) of $C$ at $\mu$. 
\end{defn}

Note that one has a linear map $\mu\mapsto\smile_\mu$, extending $\mu$ as a graded binary operation $\smile_\mu$ of $C$, so that $\mu = -\II\smile\II$. To keep notations simple, when clear from  context, we may omit the subscript $\mu$ in notations and write, e.g, $\smile_\mu\=\smile$.

\begin{exam}
For an endomorphism operad (Example \ref{HG}) $\EE_L$ one has
\[
f\u g=(-1)^{fg}\mu\o(f\t g),
      \quad \mu\t f\t g\in\EE_L^2\t\EE_L^f\t\EE_L^g.
\]
\end{exam}

\begin{prop}
Denote $\mu^{2}\=\mu\bul\mu\in C^3$. One has
\begin{align*}
f\u g=(-1)^f\braket{\mu|fg}, \quad 
(f\smile g)\smile h-f\smile(g\smile h)=\braket{\mu^{2}|fgh}.
\end{align*}
\end{prop}

Thus, in general, $\CUP C$ is a \textit{non-associative} algebra. The ternary operation $\mu^{2}\in C^{3}$ is an obstruction to associativity of $\CUP C$ and is called an \textit{associator}. The binary operation (or flow) $\mu$ is said to be \textit{associative} if $\mu^{2}=0$. This term can be explained also by the following

\begin{exam}[associator]
For the endomorphism operad $\EE_L$, one can really recognize the associator:
\[
\mu^{2}=\mu\o(\mu\t\1_L-\1_L\t\mu)\quad \in\EE_L^3, \quad \mu\in\EE_L^2.
\]
\end{exam}

The variety of binary non-associative operations $C^{2}$ may be too wide for modeling purposes. 
Every submodule $M\subset C^2$ is called an \textit{operadic (binary) model}. A model $M\subset C^2$ is called \textit{associative} if it consists of only associative operations. 

\begin{rem}
Here, lets recall essential related aspects. In 1956 Kolmogorov  proved \cite{Kol56} that every \textit{continuous} multivariate function can be presented as a finite superposition of continuous ternary and binary functions. Arnold reinforced  this result in 1959 by proving \cite{Arn59} that every  \textit{continuous ternary} function can be presented as a finite superposition of continuous \textit{binary} functions. 

Thus, restriction to the operadic \textit{binary} models is reasonable, at least when modeling the \emph{continuous} systems. 
\end{rem}

\begin{disc}[associativity vs non-associativity]
\label{associativity}
Associativity is an algebraic abstraction of the notion of symmetry. The Nature prefers symmetric  (associative) forms of existence, for \emph{rational} evolution, and its algebra of observables is \emph{a priori} associative.  Still, due to various \textit{perturbations} (forces, interactions, measurements, mutations, approximations etc),  the observable (measured) symmetries are rarely exact  but often rather deformed or broken and an associator $\mu^2$ may be considered as representing the perturbations. The binary operations near (close) \cite{ZSSS78} to associative  emerge as operadic "approximations" $\mu^2\approx0$,  representing the weak perturbations of \textit{bare} symmetries. In particular,  e.g, if (in a model) the associativity reveals in the limit $\lim_{\hbar\to0}\mu^2=0$, then such a phenomenon is called an \textit{anomaly} or \textit{quantum symmetry breaking} \cite{Bert00}. The latter is considered as a lack of a model, not the Nature.
\end{disc}

It is difficult, at least technically,  to describe all binary non-associative perturbations and certain selection rules has to be applied according to particular application. 

In Sec. \ref{OGE}, the \textit{operadic gauge equations} are justified as a selection rule for \emph{non-associative} models, the forthcoming sections may be considered as preparatory.

\section{Gerstenhaber brackets}

\begin{prop}[Getzler identity]
In an operad $C$, the Getzler identity holds:
\[
(h,f,g)
\=(h\bul f)\bul g-h\bul(f\bul g)=\braket{h|fg}+(-1)^{|f||g|}\braket{h|gf}.
\]
\end{prop}

The Getzler identity means that $\Com C$ is a non-associtive algebra, but still with a nice symmetry called the Vinberg identity.

\begin{cor}[Vinberg identity]
In $\Com C$ the  (graded right) \textit{Vinberg} identity holds,
\[
(h,f,g)=(-1)^{|f||g|}(h,g,f).
\]
\end{cor}

Thus, $\Com C$ is a (graded right) \textit{Vinberg algebra}.

\begin{thm}
\label{Albert}
If $K$ is a field of characteristic $0$, then a binary operation $\mu\in C^{2}$ generates a power-associative subalgebra in $\Com C$. 
\end{thm}
\begin{proof}
Use the Albert criterion \cite{Albert48} that a power associative algebra over a field $K$ of characteristic $0$ can be given by the identities
\[
\mu^{2}\o\mu = \mu\o\mu^{2}, \quad
(\mu^{2}\o\mu)\o\mu = \mu^{2}\o\mu^{2}.
\]
These identities easily follow from the corresponding Vinberg identities
\[
(\mu,\mu,\mu)=0, \quad
(\mu^{2},\mu,\mu)=0.
\tag*{\qed}
\]
\renewcommand{\qed}{}
\end{proof}

\begin{defn}[Gerstenhaber brackets and Jacobiator]
The \textit{Gerstenhaber brackets} $[\cdot,\cdot]$ and \textit{Jacobiator} $J$ are defined in $\Com C$ by
\begin{align*}
&[f,g]
\=[f,g]_G
\=
\braket{f|g}-(-1)^{|f||g|}\braket{g|f}, \quad 
|[\cdot,\cdot]|=0, 
%=-(-1)^{|f||g|}[g,f], 
\\
&J(f\t g\t h)
\=(-1)^{|f||h|}[[f,g],h]+(-1)^{|g||f|}[[g,h],f]+(-1)^{|h||g|}[[h,f],g],\quad |J|=0.
\end{align*}
The \textit{commutator algebra} of $\Com C$ is denoted as $\Com^{-}\!C\=(C,[\cdot,\cdot])$. One can easily see that $[\mu,\mu]=2\mu^2$.
\end{defn}

\begin{thm}[generalized Jacobi identity, cf \cite{Akivis74}] 
In $\Com^{-}C$ the generalized Jacobi identity holds,
\begin{align*}
J(f\t g\t h)
&=
  (-1)^{|f||h|}[(f,g,h)-(-1)^{|g||h|}(f,h,g)]
+(-1)^{|g||f|}[(g,h,f)-(-1)^{|h||f|}(g,f,h)]\\
&+(-1)^{|h||g|}[(h,f,g)-(-1)^{|f||g|}(h,g,f)].
\end{align*}
\end{thm}

By using the Vinberg identity on can now easily see the Lie algebra structure in $C$.

\begin{thm}
$\Com C$ is a graded Lie-admissible algebra, i.e its commutator algebra $\Com^-\!C$ is a graded Lie algebra. The Jacobi identity reads $J=0$.
\end{thm}

Thus,  the Gerstenhaber bracketing is a natural extension of the conventional Lie brackets to operadic systems. 

\begin{cor}
Define $R_f g\=[g,f]$. One has
\begin{align*}
[R_f,R_g]=R_{[g,f]},\quad
R_f [g,h] = (-1)^{|f||h|}[R_f g,h]+[g,R_f h].
\end{align*}
\end{cor}

\section{Coboundary operator}

Let $\mu\in C^2$ be a binary non-associative operation. 

\begin{defn}[cobounday operator]
Define the (pre-)\textit{coboundary} operator $\de_\mu\: C\to C$ as
$\de_\mu\=-R_\mu$, i.e
\[
-\de_\mu f  \= [f,\mu], \quad |\de_\mu|=1=|\mu|.
\]
\end{defn}

One again may omit the subscript $\mu$, when clear from the context, thus sometimes denoting $\de\=\de_\mu$.

\begin{rem}[Hochschild  coboundary operator]
For an endomorphism operad $\EE_L$ one can easily recognize the Hochschild coboundary operator as follows:
\begin{align*}
\de_\mu f=
\mu\o(\1_L\t f)
-\sum^{|f|}_{i=0}(-1)^{i}f\o\left(\1_L^{\t i}\t\mu\t\1_L^{\t(|f|-i)}\right)
+(-1)^{|f|}\mu\o(f\t\1_L).
\end{align*}
\end{rem}

\begin{prop}
One has the (right) derivation property in $\Com^-C$:
\[
\de [f,g] = (-1)^{|g|}[\de f,g] + [f,\de g].
\]
\end{prop}

\begin{prop}
One has $\de^{2}_\mu=-\de_{\mu^{2}}$. 
\end{prop}

\begin{proof}
For convenience of the reader, when calculating by assuming restriction $\ch K\neq 2$, one can find a general proof in \cite{KPS00}, 
\begin{align*}
2\de^{2}_\mu 
= [\de_\mu,\de_\mu]
=-\de_{[\mu,\mu]}
=-2\de_{\mu\o\mu}.
\tag*{\qed}
\end{align*}
\renewcommand{\qed}{}
\end{proof}

\begin{cor}
If $\mu^2=0$, then $\de_\mu^2=0$,  which in turn implies that $\IM\de_\mu\subseteq\Ker\de_\mu$. 
\end{cor}

\begin{defn}[cohomology]
Let $\mu$ be a binary \textit{associative} operation in $C$. Then the associated cohomology ($\NN$-graded module) is defined as the graded quotient  module $H_\mu(C)\=\Ker\de_\mu/\IM\de_\mu$ with homogeneous components
\[
H_\mu^{n}(C)
\=
\Ker(C^{n}\stackrel{\de_\mu}{\rightarrow}C^{n+1})/
\IM(C^{n-1}\stackrel{\de_\mu}{\rightarrow}C^{n}),
\]
where, by convention, $\IM(C^{-1}\stackrel{\de_\mu}{\rightarrow}C^{0})\=0$. Operations from $Z_\mu(C)\=\Ker \delta_\mu$ are called \textit{cocycles} and from $B_\mu(C)\=\IM\de_\mu$ \textit{coboundaries}. Thus, $H_\mu(C)\=Z_\mu(C)/B_\mu(C)$ and  the standard homological algebra technique  is applicable.
\end{defn}

\begin{rem}
For an endomoprhism operad the construction is called the Hochschild cohomology of an associative algebra.
\end{rem}

\begin{thm}
Let $\mu$ be a binary associative operation in $C$. Then the triple $(H_\mu(C),[\cdot,\cdot],\de_\mu)$ is a differential graded Lie algebra with respect to $[\cdot,\cdot]$-multiplication induced from $\Com C$.
\end{thm}

\begin{rem}[tangent cohomology and Lie theory]
By resuming at this stage, one can state that every binary associative operation $\mu\in C^{2}$ generates a graded Lie algebra $(H_\mu(C),[\cdot,\cdot])$ called the \textit{tangent cohomology} or \textit{infinitesimal algebra} of $\mu$. Its construction is strikingly natural, just as constructing the conventional Lie bracketing for matrices or the tangent Lie algebra of a Lie (transformation) group. 

\medskip
We know that the Lie bracketing is related to the matrix multiplication via the (left and right) Leibniz rules, thus it is natural to post a question about extending these rules to operadic systems.  First observe the \textit{right} Leibniz rule for cup-algebra as follows.
\end{rem}

\begin{thm}[right Leibniz rule in an operadic system]
\label{right-der}
In an operad $C$ on has the \textit{right} Leibniz rule as follows:
\[
\braket{f\u g|h}=f\u\braket{g|h}+(-1)^{|h|g}\braket{f|h}\u g.
\]
\end{thm}

\begin{rem}
An operad would be a perfectly ideal computational tool provided that also the \textit{left} Leibniz rule would hold, but, unfortunately, as we shall see, in general it does not hold, but holds only in the tangent cohomology.
\end{rem}

\section{Variations of operadic flows and Stokes law}
\label{Stokes}

Associativity of $\smile_\mu$  is clear -- it is implied by the vanishing associator $\mu^{2}=0$.
How about commutativity? To answer this and related questions we must consider \textit{variations} of the operadic flows as follows.

\begin{defn}
\label{var}
Define the \textit{variations} of some low-order simplicial operadic flows by superpositions
\begin{align*}
&\bar{\de}\braket{f|g}
\=\de\braket{f|g}-\braket{f|\de g}-(-1)^{|g|}\braket{\de f|g},\\
&\bar{\de} \braket{h|fg}
\=\de\braket{h|fg}
     -\braket{h|f\de g}
     -(-1)^{|g|}\braket{h|\de f g}
     -(-1)^{|g|+|f|}\braket{\de h|fg},\\
&\bar{\de}_\smile (f\t g) 
\= \de(f\smile g)-f\smile\delta g-(-1)^{g}\delta f\smile g.   
\end{align*}
Generalization to higher finite order flows is evident. We call $\bar{\de}$ the \textit{operadic variational operator}.
\end{defn}
It is not difficult to see that variations of the operadic flows are pairings of the ground simplexes and flows of the corresponding order. To calculate the variations one must accordingly deform the ground simplexes. As a matter of fact, these deformations can be described by using a deformed boundary operator $\bar{\partial}\leftmapsto\partial$,  see \cite{Ger63, KPS00,KP01,KP02} for more detailed exposition. 

\begin{thm}[operadic Stokes law \cite{KPS00,KP01,KP02}, cf \cite{Ger63}]
The variations of the operadic flows result as superpositions over the corresponding deformed  boundaries,
\begin{align*}
\bar{\delta} = \bar{\delta}\big|_{\bar{\partial}}
\quad \text{i.e}\quad
\bar{\delta} \braket{\cdot|\cdot} = \bar{\delta}\big|_{\bar{\partial}{\braket{\cdot}}}  \braket{\cdot|\cdot} ,\quad
\bar{\delta} \braket{\cdot|\cdot\cdot} = \bar{\delta}\big|_{\bar{\partial}{\braket{\cdot\cdot}}}  \braket{\cdot|\cdot\cdot},\quad
\bar{\delta} \braket{\cdot|\cdots} =   \braket{\cdot|\cdots},\quad
\dots
\end{align*}
\end{thm} 

\begin{rem}
One may use the operadic Stokes law do define the operadic variational operator $\bar{\de}$, then the above definition \ref{var} turns out to be a theorem to be proved.  
\end{rem}

By performing computations, the result reads as follows.

\begin{cor}[operadic Stokes law {\cite{KPS00,KP01}, cf \cite{Ger63}}]
\label{OSL} 
In an operad $C$, the low-order variations of flows read
\begin{align*}
&(-1)^{|g|}\bar{\de}\braket{f|g}=f\u g-(-1)^{fg}g\u f,\\
&(-1)^{|g|} \bar{\delta} \braket{h|fg}
=\braket{h|f}\u g+(-1)^{|h|f}f\u\braket{h|g}-\braket{h|f\u g},\\
&(-1)^{g}\bar{\de}_\smile (f\t g)=\braket{\mu^{2}|fg}.
\end{align*}
\end{cor}
The 2nd formula tells us that, really, the \textit{left} translations of $\Com C$ are \textit{not} the \textit{left} derivations of $\CUP C$. The 3rd one means that the coboundary operator $\delta$ need not be a derivation of $\CUP C$, and the associator $\mu^{2}$ again appears as an obstruction. As a corollary we can state the following.

\begin{thm}
Let $\mu$ be a binary associative operation in $C$. Then the triple $\left(H_\mu(C),\smile_\mu,\de_\mu\right)$ is a differential graded commutative associative algebra with respect to $\smile_\mu$-multiplication induced from $\CUP C$.
\end{thm}

By combining  the 2nd formula from Theorem \ref{OSL} with Theorem \ref{right-der} we obtain the operadic Stokes law in terms of the Gerstenhaber brackets as follows.

\begin{cor}[operadic Stokes law for Gerstenhaber brackets]
\label{second*} In an operad $C$, one has
\[
(-1)^{|g|}\bar{\delta}\braket{h|fg}
=
[h,f]\u g+(-1)^{|h|f}f\u[h,g]-[h,f\u g].
\]
\end{cor}

We can now state the Leibniz rule in tangent cohomology.

\begin{thm}[Leibniz rule in tangent cohomology \cite{KP01}, cf. \cite{Ger63}]
Let $\mu$ be a binary associative operation in $C$ and $h,f,g$ are homogeneous elements in the 
tangent  cohomology $H_\mu(C)$ of $\mu$. Then the left Leibniz rule holds in $H_\mu(C)$:
\begin{align*}
\braket{h|f\u g}=\braket{h|f}\u g+(-1)^{|h|f}f\u \braket{h|g},\quad
[h,f\u g]=[h,f]\u g+(-1)^{|h|f} f\u [h,g].
\end{align*}
\end{thm}

Thus, one can state that the graded commutativity and Leibniz rule in the tangent cohomology of $\mu$ are induced by  the 1st and 2nd order operadic Stokes law in $C$, respectively.  

\section{Gerstenhaber theory \& MOD I}

Now, following Gerstenhaber \cite{Ger63}, the differential calculus  in a linear operadic system  $C$ can be clarified. Select a binary \textit{associative} operation $\mu\in C^2$. Several amazingly nice coincidences happen. Due to associativity $\mu^{2}=0$ one has $\de_\mu^{2}=0$, which implies $\IM\de_\mu\subseteq\Ker\de_\mu$, hence the tangent cohomology space $H_\mu(C)\=\Ker\de_\mu/\IM\de_\mu$  is  \textit{correctly} defined as well as the triple $G_\mu(C)\=(H_\mu(C),\smile_\mu,[\cdot,\cdot])$ with two (induced) algebraic operations in  $H_\mu(C)$.  In close analogy with the conventional matrix calculus, the operation $\smile_\mu$ is associative and (graded) commutative, whereas $[\cdot,\cdot]$ is a (graded) Lie bracketing, and the both operations are related via the Leibniz rule, the latter represents the operadic Stokes law.  In a sense, one has a graded \textit{analogue} of the Poisson algebra that is everywhere used in (classical and quantum) physics. We collect the algebraic properties of the construction in a modified form as follows.

\begin{defn}[Gerstenhaber algebra \cite{Ger63}]
A \textit{Gerstenhaber algebra} is a triple triple $G\=(H,\cdot,[\cdot,\cdot])$ with the following data.
\begin{itemize}
\itemsep-2pt
\item[1)]
$H\=(H^n)_{n\in\ZZ}$ is a sequence of unital $K$-modules $H^n$. The degree of $h\in H^{n}$ is denoted by
$|h|\=n$.
\item[2)]
The pair $(H,[\cdot,\cdot])$ is a graded Lie algebra with multiplication $[\cdot,\cdot]$ of degree $|[\cdot,\cdot]|$.
\item[3)]
The pair $(H,\cdot)$ is a graded commutative associative algebra with multiplication $\cdot$ of degree $|\cdot|$
\item[4)]
The following Leibniz rule holds for homogeneous elements of $H$:
\[
[h,f\cdot g]=[h,f]\cdot g+(-1)^{(|h|+|[\cdot,\cdot]|)(|f|+|\cdot|)} f\cdot[h,g].
\]
\item[5)]
$0\neq |\cdot|-|[\cdot,\cdot]|=1$.
\end{itemize}
\end{defn}

Note that due to the last property, the Poisson algebras are not particular cases of the Gerstenhaber algebras.

\begin{defn}[tangent Gerstenhaber algebra]
Let $\mu\in C^2$ be a binary \textit{associative} operation. The triple $G_\mu(C)\=(H_\mu(C),\smile_\mu,[\cdot,\cdot])$ is called the \textit{tangent} Gerstenhaber algebra of $\mu$.
\end{defn}

The Gerstenhaber theory tells us that the Gerstenhaber algebras can be seen as \textit{classifying} objects of the binary associative operations.

\begin{exam}
(1. Hochschild cohomology) In the Hochschild complex, the Gerstenhaber algebra structure appears \cite{Ger63} in the cohomology of an associative algebra. (2. QFT) Batalin-Vilkovisky (BV) algebra \cite{BV81,BV83}.
\end{exam}

In Sec. \ref{associativity}, we already justified the associativity law $\mu^{2}=0$ as a ground  (bare) symmetry and  that the Nature prefers symmetric evolutionary forms for \textit{rational} evolution. On the other hand, we defined the operadic \textit{variational} operator in the way that variations of the operadic flows are given by superpositions of their \textit{boundary values}. The \textit{rational (observable)} operadic flows are realized by the assumption that operations are \textit{invariant} with respect to translations $\de$, i.e these turn out to be cocycles,
\begin{align*}
\de\mu=-2\mu^2=0,\quad \de h=0,\quad \de f=0,\quad \de g=0, \quad \ldots
\end{align*}
In this case, variations of the operadic  observable flows vanish (modulo coboundary) and the operadic Stokes law reads as the graded anti-commutativity and Leibniz rule in the tangent Gerstenhaber algebra of $\mu$. Lets sum the above as follows.

\begin{MOD1}[rational (cohomological) variational principle]
The operadic systems evolve in  a rational (cohomological) way. In other words, the rational (observable)  operadic flows are realized in the tangent Gerstenhaber algebras of the binary associative operations.
\end{MOD1}

MOD I is an operadic modification of the BRST quantization concept  (see e.g \cite{HT92} as well as the cohomological representation of the classical electric network  \cite{Roth55}) that the quantum physical states are realized in the BRST cohomology. When using a binary associative operation, instead of a Lie algebra (as in the BRST formalism), the (classifying) Gerstenhaber algebra structure is automatically encoded to a particular model, in the tangent cohomology of the associative operation.

\section{Deformations/perturbations \& MOD II}

For an operadic system $C$, let $\mu_0,\mu\in C^{2}$ be two binary non-associative operations. The difference $\w\=\mu_0-\mu$ is called a \textit{deformation} or \textit{perturbation} of $\mu$. 

In \textit{physical} terms, one may consider $\mu$ as a \textit{real} (true, absolute, unperturbed, bare, initial etc) operation vs its \textit{measured} (approximate, perturbed etc)  value $\mu_0$, so that $\w$ is an  operadic \textit{measurement error}.  In physical  measurement processes one may consider the perturbation (measurement error) $\w$ infinitesimally small, but it never vanish, i.e $0\neq\w\to 0$.

To cover wider area of modeling context, one may also consider $\mu$ as a \textit{ground} (unperturbed) operation representing a \textit{ground} (unperturbed) state of an operadic system vs its perturbed (measured) value $\mu_0$.  Besides of the physical measurements, the list of potential sources of the deformations includes but is not limited to
\begin{itemize}
\itemsep-2pt
\item various measurable and unmeasurable perturbations,
\item physical forces, interactions, incl., e.g,  self-interactions,
\item various approximations and estimation errors, incl., e.g, the numerical ones,
\item mutations in biophysical systems.
\end{itemize}
In what follows, it is convenient to use the standard term \textit{perturbations} for all kinds of (known or unknown) deformation sources, in particular, when working with physical models, and when non-associativity is involved as well. In pure mathematical contexts, the term \textit{deformation} is a standard accepted mathematical term.

\medskip
Tacitly assuming $\ch K\neq2$, denote the associators of $\mu$ and $\mu_0$ by 
\[
A\=\mu^{2}=\frac{1}{2}[\mu,\mu], \quad
A_0\=\mu_0^{2}=\frac{1}{2}[\mu_0,\mu_0].
\]
We call $A_0$ the \textit{deformed} (or perturbed) associator. The difference $\Omega\=A_0-A$ is called a deformation or perturbation of the associator $A$. The deformation is called \textit{quasi-associative} if $\Omega=0$ and \textit{associative} if $A=0=A_0$. Again, in a physical measurement process one may consider the perturbation $\Omega$ infinitesimally small, but it never vanish as well, i.e $0\neq\Omega\to 0$.

We already stressed (in Remark \ref{associativity}) the meaning of non-associativity as a result of deformation (perturbation) of symmetry. To find  the perturbation (deformation) equation, calculate the measured associator
\begin{align*}
A_0
&=\frac{1}{2}[\mu_0,\mu_0]\\
&=\frac{1}{2}[\mu+\w,\mu+\w]\\
&=\frac{1}{2}[\mu,\mu]
	+\frac{1}{2}[\mu,\w]
	+\frac{1}{2}[\w,\mu]
	+\frac{1}{2}[\w,\w]\\
&=A
     +\frac{1}{2}(-1)^{|\mu||\w|}[\w,\mu_0]
     +\frac{1}{2}[\w,\mu]
     +\frac{1}{2}[\w,\w]\\
&=A+[\w,\mu]+\frac{1}{2}[\w,\w]\\
&=A-\de_{\mu}\w+\frac{1}{2}[\w,\w].
\end{align*}
Denoting $d\=-\de_{\mu}$, we obtain the  \textit{perturbation (deformation)} equation called the \textit{generalized}  \cite{Paal02} \textit{Maurer-Cartan equation}:
\[
\boxed{
\Omega
\=\underbrace{A_0-A=\mu_0^{2}-\mu^{2}}_{\text{operadic perturbation}}
=\underbrace{d\w+\frac{1}{2}[\w,\w]}_{\text{operadic curvature}}
}
\]
One can see that the operadic perturbation $\Omega$, as an induced deformation of the ground associator $A$, can be seen as an operadic (form of) \textit{curvature} while the deformation $\w$ itself is (in the role of) a \textit{connection}. This observation may be fixed as follows. 

\begin{MOD2}[cf Sabinin \cite{Sab81} and \cite{NeSab00,Kikkawa64,Akivis78,Paal02}]
The associator is an \textit{operadic} equivalent of the differential geometric \textit{curvature} and may be used for representation of perturbations. 
\end{MOD2}

If the deformation $\w$ is quasi-associative, i.e $A=A_0$, we obtain the \textit{Maurer-Cartan equation} as follows:
\[
A=A_0\quad \Longleftrightarrow\quad d\w+\frac{1}{2}[\w,\w]=0.
\]
It tell us that in this case, the deformation $\w$ is itself an associative (modulo coboundary) operation,  because its associator reads $\w^{2}=-d\w$. If a quasi-associative deformation is a coboundary, i.e $\w=d\alpha$, then it is \emph{exactly} associative because it satisfies the \textit{master equation} $[\w,\w]=0$.

Following more the differential geometric analogies, now note that 
\[
-\de_{\mu_0} f \= [f,\mu_0]=[f,\mu+\w]=[f,\mu]+[f,\w] =df+[f,\w].
\]
Hence, it is natural to call $\nabla\=-\de_{\mu_0}$ a \textit{covariant derivation}. One has 
\[
\nabla f = df +[f,\w],\quad 
\nabla^{2}f=[f,A_0].
\]
Note that the condition $\nabla^{2}=0$, if applied,  does not imply that $A_0=0$, but, instead of this that $A_0$ lies in the \textit{center} of $\Com^{-}C$. In particular,
\[
\nabla^{2}=0
\quad \Longrightarrow \quad
dA_0=0
\quad \Longrightarrow \quad
A_0\in\Ker d.
\]

\section{Moduli, operadic dynamics \& MOD III}

Now, consider an associative model $M\subseteq C^{2}$. Collect the binary associative operations in such a model that are "close" to each other, expectedly (by definition) the ones with with \textit{isomorphic} tangent Gerstenhaber  algebras. The corresponding geometric picture is a graded principal fiber bundle $P(M,G,\pi)$, where the \textit{base} $M$ is called a \textit{moduli space} (of deformations, perturbations), $P$ is the \textit{total space} called a \textit{deformation} or \textit{perturbation bundle} over $M$, and $\pi:P\to M$ is the \textit{canonical projection} map, so that all fibers $\pi^{-1}(\mu)$ ($\mu\in M$), as Gerstenhaber algebras, are \textit{isomorphic} to the typical (\textit{classifying}) fiber $G\cong G_{\mu}\=\pi^{-1}(\mu)$. 

One can try to attach various structures on the deformation (perturbation) bundles as well as on the corresponding base moduli spaces, e.g, the topological and differentiable manifold structures, geometrical structures, connections, metrics etc. 

In particular, e.g, one may consider the  \textit{Batalin-Vilkovisky (BV) bundles} with BV-algebra \cite{BV81,BV83} as a typical (classifying) fiber and connections therein.

\medskip
According to contemporary geometrical interpretations of the fundamental physical interactions, the latter can be realized as connections in various principal fiber bundles. 

Thus, to involve dynamics, it is not surprising to use the following guiding principle.

\begin{MOD3}[cf Laudal \cite{Laudal11}]
\label{MOD3}
The \textit{time} (\textit{interval} between the space-time events) is a \textit{measure of change} (deformation, perturbation) i.e a metrics on the moduli space. Connections in the deformation bundles over the moduli spaces   represent perturbations. 
 \end{MOD3}

MOD III means that the physical space-time is presented (modeled) as a moduli space of associative deformations and operations from the tangent Gerstenhaber algebras are (local) \textit{operadic observables}.  In other words, one can also say that the space-time is covered by the \textit{cohomology fields} \cite{KP96}, and their dynamics (change in time)  must be described. The dynamics of an operadic system may be represented by isomorphisms $G\to G_\mu$ where the binary associative operation $\mu$ becomes a dynamical variable, representing the time.

\medskip
One can parametrize deformations of operations by \textit{operadic Lax representations} \cite{Paal07} of the dynamical (Hamiltonian) systems. Then the corresponding moduli spaces appear as the configuration and phase spaces of the classical dynamical and Hamiltonian systems. Some elaborated examples are presented in \cite{PV08,PV09}. In particular, the time evolution of a quantum operadic flow $f$ may be prescribed by the \emph{operadic Heisenberg equation} \cite{Paal07,Paal13} 
\begin{align*}
\frac{\hbar}{i}\frac{\partial f}{\partial t} = [h,f]_G, \quad h\t f \in H_\mu^{1}(C)\t C^{f}.
\end{align*}
The latter tells us that if the first cohomology space of a binary associative operation $\mu$ is trivial, i.e $H_\mu^{1}(C)=0$, then the operadic system $C$ is \textit{static} with respect to $\mu$.

Other aspect concerns appearance and description of non-associativity generated by the \textit{infinitesimal} perturbations. Then, the operadic approximations $\Omega^{2}\approx0$ must be taken into account. First of all, when $0\neq\mu\approx0$, the tangent cohomology space $H_\mu(C)$ of operadic observables will start to \textit{decay}  which reveals in structural changes of an operadic system during the evolutionary or measurement process as a result of the non-associative perturbations - one is a witness of the \textit{dynamical decay of the tangent cohomologies}. To handle such phenomena mathematically, one needs equations to describe the operadic curvature $\Omega$. As soon as the \emph{infinitesimal} perturbations are governed by the quantum laws, the corresponding operadic Heisenberg equation must hold,
\begin{align*}
\frac{\hbar}{i}\frac{\partial\Omega}{\partial t} = [h,\Omega]_G, \quad h\t \Omega \in H_\mu^{1}(C)\t C^{3}.
\end{align*}
In what follows, additional prescriptions are proposed, which follow the gauge theoretic laws of the Yang-Mills type \cite{YM54,HT92,YM50}.

\section{Operadic gauge equations}
\label{OGE}

Let $\w\=\mu_0-\mu\subset C^{2}$ be a \textit{non-associative} deformation. We  assume that the perturbed associator $A_0\=\mu_0^{2}\in C^{3}$ is established by the measurement processes, as an approximation of $A$. The aim its to state an equation for the measurement error $\Omega\=A_0-A$.

We follow the standard differential geometric considerations. Assume that $\ch K\neq2,3$ and differentiate the deformation equation,
\begin{align*}
d\Omega 
&=\p^{2}\w+\frac{1}{2}\p[\w,\w]\\
&=\p^{2}\w+\frac{1}{2}(-1)^{|\p||\w|}[\p\w,\w]
     +\frac{1}{2}[\w,\p\w]\\
&=\p^{2}\w-\frac{1}{2}[\p\w,\w]
     +\frac{1}{2}[\w,\p\w]\\
&=\p^{2}\w-\frac{1}{2}[\p\w,\w]
     -\frac{1}{2}(-1)^{|\p\w||\w|}[\p\w,\w]\\
&=\p^{2}\w-[\p\w,\w].
\end{align*}
Again using the deformation equation, we obtain
\begin{align*}
\p\Omega
&=\p^{2}\w-[\p\w,\w]\\
&=\p^{2}\w-[\Omega-\frac{1}{2}[\w,\w],\w]\\
&=\p^{2}\w-[\Omega,\w]+\frac{1}{2}[[\w,\w],\w].
\end{align*}
It follows from the Jacobi identity that $[[\w,\w],\w]=0$.
Hence, 
\[
\p \Omega=\p^{2}\w-[\Omega,\w].
\]
Now recall that $\p^{2}=-\p_{A}$ ($\neq0$, in general) and one can see that
\begin{align*}
\nabla \Omega 
&\= \p \Omega+[\Omega,\w] \\
&=-\p_{A}\w\\
&=-[A,\w]\\
&=-dA-[A,\w]\\
&=-\nabla A.
\end{align*}
Thus, the perturbed associator $A_0$ satisfies the operadic differential equation called the operadic \textit{Bianchi identity}
\[
\nabla A_0\=dA_0+[A_0,\w]=0.
\]
Due to $\nabla^{2}A_0=[A_0,A_0]=0$, one can see that further differentiation of the Bianchi identity does not produce additional constraints. 

To clarify the (algebraic) meaning of the Bianchi identity and its solvability, note that
\[
0=\p A_0+[A_0,\w]=[A_0,\mu]+[A_0,\mu_0-\mu]=-[A,\mu]=[\mu^{2},\mu]=\mu^{2}\o\mu-\mu\o\mu^{2},                                                                                                                                                                   
\]
so the Bianchi identity strikingly reads as a power-associativity constraint for the composition multiplication $\o:=\braket{\cdot|\cdot}$ from Theorem \ref{Albert} (the Albert criterion),
\[
\mu^{2}\o\mu=\mu\o\mu^{2}.
\]
Hence, the set of all solutions (the general solution) of the operadic Bianchi identity is extremely wide, consisting of all  binary \textit{non-associative} operations in $C$, i.e the whole $C^{2}\subset C$. Certainly, further restriction to   submodules of $C^{2}$ is sensible. 

\medskip
Now, apply the gauge theoretic prescriptions (laws) of the Yang-Mills type \cite{YM54,HT92,YM50}. It is well known that the geometric part of the gauge field equations are presented as a Bianchi identity, thus using the above operadic Bianchi identity $\nabla\Omega=-\nabla A$ is natural. Meanwhile, accepting the guiding modeling principle MOD I, we believe that at least the unperturbed (bare) symmetry is exact, i.e the associativity law $\mu^{2}=A=0$ holds;  the non-associativity $\Omega\neq0$ may only come into play from physical measurements and other  perturbations. Hence, using the operadic Bianchi identity $\nabla\Omega=0$ is natural. 

\medskip
To introduce the second operadic equation for $\Omega$, one must involve a model restricting \textit{operadic "dual"} $\Omega^{\dag}\leftmapsto\Omega$ and an  \textit{operadic current} $\mathcal{J}\in C^{\Omega^{\dag}+1}$. Then, the  \textit{operadic gauge equations} (cf \cite{Paal02}) for a non-associative perturbation $\Omega\in C^{3}$ of the Yang-Mills type read
\[
\boxed{
\nabla \Omega\=d\Omega+[\Omega,\w]=0,\quad
\nabla \Omega^{\dag}\=d\Omega^{\dag}+[\Omega^{\dag},\w]=\mathcal{J}, \quad \text{where } d^{2}=0
}
\]
Note that $\nabla\mathcal{J}=[\Omega^{\dag},\Omega]\in C^{\Omega^{\dag}+2}$. Thus, the natural constraint $[\Omega^{\dag},\Omega]\=0$ is equivalent to the \textit{operadic conservation law} $\nabla\mathcal{J}=0$.

\medskip
Today, there is not much experience in handling the operadic differential equations, except the Maurer-Cartan one  in the algebraic deformation theory. Most probably, the operadic dual $\Omega^{\dag}$ can be further specified by following the Laudal principle (see MOD III, Sec. \ref{MOD3}), at least for the \textit{infinitesimal} perturbations $\Omega\approx0$, i.e for the models near to associative, by elaborating the operadic  approximation methods with  the operadic Heisenberg equation of $\Omega$. The standard \textit{correctness} conditions may be applied in the operadic modeling as well  -  the existence, uniqueness and stability of solutions of the operadic differential equations and also renormalization principles.  

\medskip
Finally, note that the operadic  (anti-)\textit{self-dual} models with \textit{ansatz} $\Omega^{\dag}=\pm \Omega$ (then evidently $\mathcal{J}=0$) can be, in general, generated by binary non-associative operations, i.e by the whole of $C^{2}$, which may be too wide, whereas the (quasi-)associativity constraint $\Omega=0$ may be too restrictive. Instead, one may consider \emph{weakly non-associative} models with \textit{approximate} (anti-)self-dual \textit{ansatz} $\Omega^{\dag}\approx\pm\Omega$ and \textit{weak} operadic current $\mathcal{J}\approx0$. It explains the importance of the non-associative models, in particular, the ones \textit{near} to associative, for physics as well as for other natural sciences and applications.  

%\section*{Acknowledgement}
\medskip
\footnotesize{The research was in part supported by the Estonian Research Council, Grant ETF9038. The author is grateful to the referees for useful remarks and corrections.}

\par\medskip
\noindent
\footnotesize{
Tallinn University of Technology, Estonia
}
\end{document}